\documentclass[11pt]{amsart}
\usepackage{amsmath,amstext,amsthm,amsfonts,amssymb,amsthm}
\usepackage{multicol}
\usepackage{latexsym}
\usepackage{color}
\usepackage{graphicx}
\usepackage{epsf}
\usepackage{epsfig}
\usepackage{subfigure}
\usepackage{rotating}
\usepackage{curves}
\usepackage{epic}

\usepackage{graphics}
\usepackage{verbatim}
\usepackage{color}
\usepackage{hyperref}

\newtheorem{theorem}{Theorem}[section]
\newtheorem{lemma}[theorem]{Lemma}
\theoremstyle{definition}
\newtheorem{definition}[theorem]{Definition}
\newtheorem{proposition}[theorem]{Proposition}

\theoremstyle{remark}

\numberwithin{equation}{section}
\newcommand{\abs}[1]{\lvert#1\rvert}

\newcommand{\N}{\mathcal{N}}

\newcommand{\HI}{\mathfrak{H}}

\newcommand{\R}{\mathbb{R}}

\newcommand{\D}{\mathfrak{D}}
\newcommand{\C}{\mathbb{C}}
\newcommand{\HQ}{\mathbb{H}}

\newcommand{\oz}{\overline{z}}
\newcommand{\qu}{\mathbf{q}}

\begin{document}
\title[Quaternionic coherent states]{Regular subspaces of a quaternionic Hilbert space from quaternionic Hermite polynomials and associated Coherent states}
\author{K. Thirulogasanthar$^{1}$ and S. Twareque Ali$^{2}$}
\address{$^{1}$ Department of Computer Science and Software
Engineering, Concordia University, 1455 De Maisonneuve Blvd. West,
Montreal, Quebec, H3G 1M8, Canada. }
\address{$^2$ Department of mathematics and Statistics, Concordia University, Montreal, Quebec, H3G 1M8, Canada}
\email{santhar@gmail.com and stali@mathstat.concordia.ca}
\thanks{The research of one of the authors (STA) was supported by the Natural Science and Engineering Research Council of Canada (NSERC).}
\date{\today}
\begin{abstract}
We define quaternionic Hermite polynomials by analogy with two families of complex Hermite polynomials.  As in the complex case, these polynomials consatitute  orthogonal families of vectors in ambient quaternionic $L^2$-spaces. Using these polynomials, we then define regular and anti-regular subspaces of these  $L^2$-spaces, the associated reproducing kernels and the ensuing quaternionic coherent states.
\end{abstract}

\maketitle
\pagestyle{myheadings}

\section{Introduction}\label{sec_intro}
Building quantum mechanics on quaternionic Hilbert spaces has been a much studied problem for many
years (see, for example, \cite{Ad} and the many references cited therein). Associated to
this problem is that of building appropriate families of coherent states on quaternionic
Hilbert spaces. The fact
that the analogues of the usual canonical coherent states cannot be built using a group theoretical
 argument in the case of quaternions, has been elaborated in \cite{AdMil}. On the other hand,
 analogues of such coherent states in a quaternionic setting have been constructed using other
 methods in \cite{AliBhatRoy} and \cite{Thi2}. In this paper we study the possibility of
 constructing some analogues of the so-called {\em non-linear} coherent states on quaternionic
 Hilbert spaces, using the recently developed holomorphic function theory for quaternionic
 variables \cite{Am,Gra1,Gra2}.

Recall that the real Hermite polynomials are defined by
\begin{equation}\label{I1}
H_n(x)=(-1)^ne^{x^2}\frac{d^n}{dx^n}e^{-x^2}
\end{equation}
and it is well known that the functions $C_ne^{-x^2/2}H_n(x)$,  for some normalization constants
$C_n$, are the eigensolutions of the quantum harmonic oscillator  \cite{AFG}. As an immediate
extension of the real Hermite polynomials, by replacing the real variable $x$ by a complex
number $z$, the complex Hermite polynomials, $H_n(z)$, were studied in \cite{VE}. In particular,
in \cite{VE}, it has been  shown that these complex Hermite polynomials form an  orthonormal
basis of a certain Hilbert space of complex functions over $\mathbb C$ and  this Hilbert space
is a reproducing kernel Hilbert space. In \cite{Ga}, using these holomorphic Hermite polynomials,
a set of coherent states (CS) have been  built, which is then used to study  some quantum
mechanical issues and a quantization of the non-commutative plane. Apart from $H_n(z)$, another
interesting set of Hermite polynomials, $H_{n,m}(z,\overline{z})$, were studied in
\cite{Ghan,Mat,Wil} (see also references therein). In these papers it has been shown that
the functions, $e^{-|z|^2/2}H_{n,m}(z,\overline{z})$ are eigensolutions of the Landau problem
\cite{AFG,Ghan}. Recently, several interesting features of the $H_{n,m}(z,\overline{z})$ have
been studied by fixing either $n=0$ or $m=0$. In fact, by so fixing one can recover the
holomorphic and anti-holomorphic subspaces of a certain $L^2$-space and the subspaces so
obtained can also be identified as the well known Bargmann spaces of holomorphic and
anti-holomorphic functions. These are also reproducing kernel Hilbert spaces with reproducing
kernels associated to the canonical coherent states. More generally, other reproducing kernel
Hilbert spaces have been obtained using  subsets of  $H_{n,m}(z,\overline{z})$ as bases
\cite{AFG,Nic}, which also admit coherent states. In \cite{Nic}, these coherent states have
been used to implement quantizations of $\mathbb C$.
In \cite{Nic,Ga} the Hermite polynomials $H_{n,m}(z,\overline{z})$ and $H_n(z)$ and kernels
associated with these polynomials have been used to obtain coherent states and the authors
have used the CS so obtained to study some quantum phenomena and quantizations. In fact, the
procedure used in
\cite{Nic,Ga} to build CS was earlier  worked out in \cite{Iwa} and later in \cite{Ali, Thi3}
as generalization of the definition of canonical CS.
In this paper we shall use a similar approach to obtain quaternionic coherent states.

 For the sake of completeness we briefly revisit the procedure for building generalized and
 nonlinear coherent states.

Let $\{\phi_m\}_{m=0}^{\infty}$ be an orthonormal basis of an
abstract separable Hilbert space $\HI$. The well known canonical
coherent states are defined by:
\begin{equation}
\mid z\rangle=e^{-\frac{r^2}{2}}
\sum_{m=0}^{\infty}\frac{z^m}{\sqrt{m!}}\phi_m\in\HI,
\label{can-comp_CS}
\end{equation}
where $z = re^{i\theta}\in\mathbb C$, the complex plane.

A possible generalization of the above definition of canonical coherent states, to the so-called nonlinear coherent states, goes as  follows: Let $\D$ be an open subset of $\C$. For $z\in\D$ set
\begin{equation}\label{can-cs}
\mid z\rangle=\N(|z|)^{-\frac{1}{2}}
\sum_{m=0}^{\infty}\frac{z^m}{\sqrt{\rho(m)}}\phi_m\in\HI,
\end{equation}
where $\{\rho(m)\}_{m=0}^{\infty}$ is a positive sequence of real
numbers and $\N(|z|)$ is the normalization factor ensuring that
$\langle z\mid z\rangle=1.$ If in addition $\{\mid z\rangle \mid z\in
\D\}$ satisfy
\begin{equation}\label{can-res}
\int_{\D}\mid z\rangle\langle z\mid d\mu=I_{\HI},
\end{equation}
where $d\mu$ is an appropriately chosen measure on $\D$ and
$I_{\HI}$ is the identity operator on $\HI$, then
$\{\mid z\rangle \mid z\in \D\}$ is said to be a set of nonlinear coherent
states on $\D$. \\

More generally, (generalized) CS can be constructed as follows: Let $(\Omega,\mu)$
be a measure space and $\HQ$  a closed subspace of
$L^2(\Omega,\mu)$. Let $\{\Phi_m\}_{m=0}^{\text{dim}{(\HQ)}}$,
$\text{dim}{(\HQ)}$ denoting  the dimension of $\HQ$, be an
orthonormal basis of $\HQ$ satisfying:
$$\sum_{m=0}^{\text{dim}{(\HQ)}}\abs{\Phi_m(x)}^2<\infty$$
for all $x\in\Omega$. Let $\HI$ be another Hilbert space such that
$\text{dim}(\HQ)=\text{dim}(\HI)$. Let
$\{\phi_m\}_{m=0}^{\text{dim}{(\HI)}}$ be an orthonormal basis of
$\HI$. Define
\begin{equation}\label{rep}
K(x,y)=\sum_{m=0}^{\text{dim}{(\HQ)}}\overline{\Phi_m(x)}\Phi_{m}(y).
\end{equation}
Then $K(x,y)$ is a reproducing kernel, that is, $K(x,y)$ satisfies
\begin{enumerate}
\item[(a)] hermiticity, $K(x,y)=\overline{K(y,x)}$ for all $x,y\in \Omega$;
\item[(b)]positivity, $K(x,x)\geq 0$ for all $x\in\Omega$;
\item[(c)]idempotence,
$\int_{\Omega}K(x,y)K(y,z)d\mu(y)=K(x,z),$
\end{enumerate}
 and $\HQ$ is the
corresponding reproducing kernel Hilbert space. For $x\in \Omega$,
define
\begin{equation}\label{rep-cs}
\mid x\rangle=K(x,x)^{-\frac{1}{2}}\sum_{m=0}^{\text{dim}{(\HQ)}}\overline{\Phi_m(x)}\phi_m \in \mathfrak H\; .
\end{equation}
Then,
$$\langle x\mid x\rangle=K(x,x)^{-1}
\sum_{m=0}^{\text{dim}{(\HQ)}}\overline{\Phi_m(x)}\Phi_m(x)=1,$$
and
$$\mathcal{W}:\HI\longrightarrow\HQ\;\;\;\text{with}\;\;\;\mathcal{W}\phi(x)
=K(x,x)^{\frac{1}{2}}\langle x\mid \phi\rangle$$
is an isometry. Then, for $\phi,\psi\in\HI$ we have
\begin{eqnarray*}
\langle\phi\mid\psi\rangle_{\HI}=\langle\mathcal{W}\phi\mid\mathcal{W}\psi\rangle_{\HQ}&=&
\int_{\Omega}\overline{\mathcal{W}\phi(x)}\mathcal{W}\psi(x)d\mu(x)\\
&=&\int_{\Omega}\langle\phi\mid x\rangle\langle x\mid\psi\rangle K(x,x)d\mu(x),
\end{eqnarray*}
and
\begin{equation}\label{res-rep}
\int_{\Omega}\mid x\rangle\langle x\mid K(x,x)d\mu(x)=I_{\HI},
\end{equation}
where $K(x,x)$ plays the role of a positive weight function. Thus, the set of
states $\{\mid x\rangle \mid x\in\Omega\}$ forms a set of (generalized) CS.

\bigskip

In the case where $\{\Phi_m\}_{m=0}^{\text{dim}{(\HQ)}}$ is an
orthogonal basis of $\HQ$, one can define
$\rho(m)=\|\Phi_m\|^2;\;m=0,1,2, ...,\text{dim}{(\HQ)},$ and obtain an
orthonormal basis
$$\left\{\frac{\Phi_m}{\sqrt{\rho(m)}}\right\}_{m=0}^{\text{dim}{(\HQ)}}$$
of $\HQ$. Then, setting
\begin{equation}\label{rep-cs-2}
\mid x\rangle=K(x,x)^{-\frac{1}{2}}
\sum_{m=0}^{\text{dim}(\HI)}\frac{\Phi_m(x)}{\sqrt{\rho(m)}}\phi_m\in\HI,
\end{equation}
one obtains the desired result which is analogous to (\ref{can-cs}).

\bigskip

The above discussion motivates the following definition.
\begin{definition}\label{def-CS}
 Let $\D$ be an open subset of
$\mathbb C$. Let
$$\Phi_m: \D\longrightarrow \C,\;\;\;m=0,1,2,\dots,$$
be a sequence of complex functions. Define
\begin{equation}\label{ma-cs}
\mid z\rangle=\mathcal N(|z|)^{-\frac{1}{2}}
\sum_{m=0}^{\infty}\frac{\Phi_{m}(z)}{\sqrt{\rho(m)}}
\phi_{m}\in\mathfrak{H};\;\;z\in\D,
\end{equation}
where $\mathcal N(|z|)$ is a normalization factor and
$\{\rho(m)\}_{m=0}^{\infty}$ is a sequence of nonzero positive
real numbers. The set of vectors in (\ref{ma-cs}) is said to form
a set of CS if
\begin{enumerate}
\item[(a)] $\langle z\mid z\rangle=1$ for all $z\in\D$;
\item[(b)]
the states $\{\mid z\rangle \mid z\in\D\}$ satisfy a resolution of the identity:
\begin{equation}
\int_{\mathfrak D}\mid z\rangle \langle z\mid d\mu=I_{\mathfrak H}
\label{2},
\end{equation}
\end{enumerate}
where $d\mu$ is an appropriately chosen measure and $I_{\mathfrak H}$ is the
identity operator on $\mathfrak H$.
\end{definition}
In \cite{Thi}, we studied what were called quaternionic CS in a complex Hilbert space. These
were  defined as vector CS, built using a quaternionic variable. However, recently, in \cite{Thi2}
we have defined the cannonical quaternionic CS by replacing the $z$ in (\ref{can-cs}) by a
quaternion $\qu$ and considered these CS as vectors in  a quaternionic Hilbert space
(see also \cite{AdMil} for Prelemov type quaternionic CS). Further, we have also defined the
quaternionic version of Hermite polynomials by replacing $z$ in $H_{n,m}(z,\overline{z})$ by a
quaternion. Since there is more than one way of defining a derivative with respect to a
quaternionic variable \cite{Am}, in \cite{Thi2} the quaternionic derivative was introduced
in a formal sense. However, recent developments of quaternionic analysis give a definition of
a quaternionic derivative, the so-called Cullen derivative \cite{Gra1,Gra2}, which is more
useful for our purposes and we adopt it here.

The novelty of the present paper can be summarized as follows: using the Cullen derivative we
define the quaternionic counterparts of the Hermite polynomials $H_n(z)$ and
$H_{n,m}(z,\overline{z})$, as vectors in  quaternionic Hilbert spaces  by replacing $z$ by
a quaternion $\qu$. The two index quaternionic Hermite polynomials $H_{n,m}(\qu,\overline{\qu})$
will span a subspace $\mathfrak H_\qu$ of a quaternionic $L^2$-space. By fixing $n=0$ or $m=0$
in $H_{n,m}(\qu,\overline{\qu})$ we shall obtain the regular and anti-regular subspaces of
$\mathfrak H_\qu$  (as the counterparts of holomorphic and anti-holomorphic subspaces). We shall
also prove that the single indexed quaternionic Hermite polynomials $H_n(\qu)$ and
$H_n(\overline{\qu})$ serve as bases for certain regular and anti-regular subspaces respectively.
Apart from these, by defining kernels with $H_n(\qu)$ and $H_{n,m}(\qu,\overline{\qu})$ we will
obtain CS resembling (\ref{ma-cs}) over quaternionic Hilbert spaces and realize the regular
subspace as a reproducing kernel Hilbert space.
The functions $e^{-|\qu|^2/2}H_{n,m}(\qu,\overline{\qu})$ are shown to be eigenfunctions of an
operator $\mathfrak{L}_H$ with infinite degeneracy as in the Landau problem \cite{Thi2} (see
also end of Section \ref{sec-quherm-polyn}). However, we are unable at this point to give a
physical meaning to $\mathfrak{L}_H$. Further investigation of this point, quantization of the
type done  in \cite{Nic,Ga} and a study of the  modular structures along the lines of
\cite{AFG}, using quaternionic coherent states, are left for future work.

\section{Mathematical preleminaries}
In order to make the paper self-contained, we recall a few facts about quaternions which may not be well-known. In particular, we revisit the $2\times 2$ complex matrix representations of quaternions, quaternionic Hilbert spaces, their duals, the Cullen derivative and the definition of regularity of a function of a quaternionic variable.

\subsection{Quaternions}
Let $H$ denote the field of quaternions. Its elements are of the form $\qu=x_0+x_1i+x_2j+x_3k$ where $x_0,x_1,x_2$ and $x_3$ are real numbers, and $i,j,k$ are imaginary units such that $i^2=j^2=k^2=-1$, $ij=-ji=k$, $jk=-kj=i$ and $ki=-ik=j$. The quaternionic conjugate of $\qu$ is defined to be $\overline{\qu} = x_0 - x_1i - x_2j - x_3k$. We shall find it convenient to use the representation of quaternions by $2\times 2$ complex matrices:
 \begin{equation}
\qu = x_0 \sigma_{0} + i \underline{x} \cdot \underline{\sigma},
 \end{equation}
with $x_0 \in \mathbb R , \quad \underline{x} = (x_1, x_2, x_3)
\in \mathbb R^3$, $\sigma_0 = \mathbb{I}_2$, the $2\times 2$ identity matrix, and
$\underline{\sigma} = (\sigma_1, -\sigma_2, \sigma_3)$, where the
$\sigma_\ell, \; \ell =1,2,3$ are the usual Pauli matrices. The quaternionic imaginary units are identified as, $i = \sqrt{-1}\sigma_1, \;\; j = -\sqrt{-1}\sigma_2, \;\; k = \sqrt{-1}\sigma_3$. Thus,
 \begin{equation}
\qu = \left(\begin{array}{cc}
x_0 + i x_3 & -x_2 + i x_1 \\
x_2 + i x_1 & x_0 - i x_3
\end{array}\right) \qquad \text{and} \qquad \overline{\qu} = \qu^\dag\quad \text{(matrix adjoint)}\; .
 \label{q3}
 \end{equation}
Introducing the polar coordinates:
 \begin{eqnarray*}
x_0 &=& r \cos{\theta}, \\
x_1 &=& r \sin{\theta} \sin{\phi} \cos{\psi}, \\
x_2 &=& r \sin{\theta} \sin{\phi} \sin{\psi}, \\
x_3 &=& r \sin{\theta} \cos{\phi},
 \end{eqnarray*}
where $r \in [0,\infty)$, $\theta, \phi \in [0,\pi]$, and $\psi
\in [0,2\pi)$, we may write
 \begin{equation}
\qu = A(r) e^{i \theta \sigma(\widehat{n})}
 \label{q4},
 \end{equation}
where
 \begin{equation}
A(r) = r\mathbb \sigma_0 \quad \text{and} \quad
\sigma(\widehat{n}) = \left(\begin{array}{cc}
\cos{\phi} & \sin{\phi} e^{i\psi} \\
\sin{\phi} e^{-i\psi} & -\cos{\phi}
\end{array}\right).
\label{q5}
 \end{equation}
The matrices
$A(r)$ and $\sigma(\widehat{n})$ satisfy the conditions,
 \begin{equation}
A(r) = A(r)^\dagger \quad,\quad \sigma(\widehat{n})^2 = \sigma_0
\quad,\quad \sigma(\widehat{n})^\dagger = \sigma(\widehat{n})
\quad,\quad \lbrack A(r), \sigma(\widehat{n}) \rbrack = 0.
 \label{san1}
 \end{equation}
Note that $\vert\qu\vert^2  := \overline{\qu} \qu = r^2 \sigma_0 = (x_0^2 +  x_1^2 +  x_2^2 +  x_3^2)\mathbb I_2$ defines a real norm on $H$.

Using the above complex representation for quaternions, we defined a set of vector coherent states (VCS) in \cite{Thi}. To recall that construction briefly, if $\{\phi_m\}_{m=0}^\infty $ is an orthonormal basis of an
abstract, complex separable Hilbert space $\mathfrak{H}$ and $\{\chi^1, \chi^2\}$ is
an orthonormal basis of $\mathbb{C}^2$, then the VCS take the form,
 \begin{equation}
\mid \qu,j \rangle = \mathcal{N}_1(\vert \qu \vert)^{-1/2}
\sum_{m=0}^\infty \frac{\qu^m}{\sqrt{x_m!}} \chi^j \otimes \phi_m
\in \mathbb C^2 \otimes \mathfrak H, \quad j=1,2,
 \label{qv2}
 \end{equation}
where $\mathcal{N}_1(\vert \qu \vert)$ and $x_m! = \rho(m)$ can be chosen appropriately. Using (\ref{san1}) we can
determine the normalization constant $\mathcal N_1(\vert \qu
\vert)$, and the resolution of the identity (for details see
\cite{Thi}). First note that, in order for the norm of the vector
$\mid \qu,j \rangle$ to be finite, we must have,
 \begin{equation}
\langle \qu,j \mid \qu,j \rangle = \mathcal N_1(\vert \qu
\vert)^{-1} \sum_{m=0}^\infty \frac {r^{2m}}{x_m!} < \infty.
 \end{equation}
Thus, if $\lim_{m \rightarrow \infty} x_m = x$, we need to
restrict $r$ to $0 \leq r < L = \sqrt{x}$ for the convergence of
the above series. In this case, define
 \begin{equation}\label{domain}
\mathcal D = \{(r, \theta, \phi, \psi) \mid 0 \leq r < L, \; 0
\leq \phi \leq \pi, \; 0 \leq \theta, \psi < 2\pi\},
 \end{equation}
The resolution of the identity is then given with respect to a measure of the type
\begin{equation}\label{measure}
d\varsigma(r, \theta,
\phi, \psi)= d\tau(r) d\theta d\Omega(\phi ,\psi ), \quad \text{ with} \quad
d\Omega(\phi ,\psi) = \displaystyle{\frac{1}{4\pi}} \sin{\phi}
d\phi d\psi .
\end{equation}

\subsection{Quaternionic Hilbert spaces}
In this subsection we  define left and right quaternionic Hilbert spaces. For details we refer the reader to \cite{Ad}. We also define the Hilbert space of square integrable functions on quaternions based on \cite{Vis}.
\subsubsection{Right Quaternionic Hilbert Space}
Let $V_{H}^{R}$ be a linear vector space under right multiplication by quaternionic scalars (again $H$ standing for the field of quaternions).  For $f,g,h\in V_{H}^{R}$ and $\qu\in H$, the inner product
$$\langle\cdot\mid\cdot\rangle:V_{H}^{R}\times V_{H}^{R}\longrightarrow H$$
satisfies the following properties
\begin{enumerate}
\item[(i)]
$\overline{\langle f\mid g\rangle}=\langle g\mid f\rangle$
\item[(ii)]
$\|f\|^{2}=\langle f\mid f\rangle>0$ unless $f=0$, a real norm
\item[(iii)]
$\langle f\mid g+h\rangle=\langle f\mid g\rangle+\langle f\mid h\rangle$
\item[(iv)]
$\langle f\mid g\qu\rangle=\langle f\mid g\rangle\qu$
\item[(v)]
$\langle f\qu\mid g\rangle=\overline{\qu}\langle f\mid g\rangle$
\end{enumerate}
where $\overline{\qu}$ stands for the quaternionic conjugate. We assume that the
space $V_{H}^{R}$ is complete under the norm given above. Then,  together with $\langle\cdot\mid\cdot\rangle$ this defines a right quaternionic Hilbert space, which we shall assume to be separable. Quaternionic Hilbert spaces share most of the standard properties of complex Hilbert spaces. In particular, the Cauchy-Schwartz inequality holds on quaternionic Hilbert spaces as well as the Riesz representation theorem for their duals (see below).  Thus, the Dirac bra-ket notation
can be adapted to quaternionic Hilbert spaces:
$$\mid f\qu\rangle=\mid f\rangle\qu,\hspace{1cm}\langle f\qu\mid=\overline{\qu}\langle f\mid\;, $$
for a right quaternionic Hilbert space, with $\vert f\rangle$ denoting the vector $f$ and $\langle f\vert$ its dual vector.

\subsubsection{Left Quaternionic Hilbert Space}
Let $V_{H}^{L}$ be a linear vector space under left multiplication by quaternionic scalars.  For $f,g,h\in V_{H}^{L}$ and $\qu\in H$, the inner product
$$\langle\cdot\mid\cdot\rangle:V_{H}^{L}\times V_{H}^{L}\longrightarrow H$$
satisfies the following properties
\begin{enumerate}
\item[(i)]
$\overline{\langle f\mid g\rangle}=\langle g\mid f\rangle$
\item[(ii)]
$\|f\|^{2}=\langle f\mid f\rangle>0$ unless $f=0$, a real norm
\item[(iii)]
$\langle f\mid g+h\rangle=\langle f\mid g\rangle+\langle f\mid h\rangle$
\item[(iv)]
$\langle \qu f\mid g\rangle=\qu\langle f\mid g\rangle$
\item[(v)]
$\langle f\mid \qu g\rangle=\langle f\mid g\rangle\overline{\qu}$
\end{enumerate}
Again, we shall assume that the space $V_{H}^{L}$ together with $\langle\cdot\mid\cdot\rangle$ is a separable Hilbert space. Also,
\begin{equation}\label{leftcs}
\mid \qu f\rangle=\mid f\rangle\overline{\qu},\hspace{1cm}\langle \qu f\mid=\qu\langle f\mid.
\end{equation}
Note that, because of our convention for inner products, for a left quaternionic Hilbert space, the bra vector $\langle f\mid$ is to be identified with the vector itself, while the ket vector $\mid f \rangle$ is to be identified with its dual.
(There is a natural left multiplication by quaternionic scalars on the dual of a right quaternionic Hilbert space and a similar right multiplication on the dual of a left quaternionic Hilbert space.)

Separable quaternionic Hilbert spaces admit countable orthonormal bases. Let $V_H^L$ be a left quaternionic Hilbert space and let $\{e_v\}_{\nu = 0}^N$  ($N$ could be finite or infinite) be an orthonormal basis for it. Then, $\langle e_\nu \mid e_\mu\rangle = \delta_{\nu \mu}$ and any vector $f \in V_H^L$ has the expansion $f = \sum_\nu f_\nu e_\nu$, with $f_\nu = \langle f\mid e_\nu\rangle \in H$. Using such a basis, it is possible to introduce a multiplication from the right on $V_H^L$ by elements of $H$. Indeed, for $f \in V_H^L$ and $\qu\in H$ we define,
\begin{equation}
   f\qu = \sum_\nu (f_\nu \qu)e_\nu .
\label{rightmult}
\end{equation}

The field of quaternions $H$ itself can be turned into a left quaternionic Hilbert space by defining the inner product $\langle \qu \mid \qu^\prime \rangle = \qu \qu^{\prime\dag} = \qu\overline{\qu^\prime}$ or into a right quaternionic Hilbert space with  $\langle \qu \mid \qu^\prime \rangle = \qu^\dag \qu^\prime = \overline{\qu}\qu^\prime$.
\subsubsection{Quaternionic Hilbert Spaces of Square Integrable Functions}
Let $(X, \mu)$ be a measure space and $H$  the field of quaternions, then
$$L^2_H(X,\mu)=\left\{f:X\rightarrow H\;\; \left| \;\; \int_X|f(x)|^2d\mu(x)<\infty \right.\right\}$$\label{L^2}
is a left quaternionic Hilbert space, with the (left) scalar product
\begin{equation}
\langle f \mid g\rangle =\int_X f(x)\overline{g(x)} d\mu(x),
\label{left-sc-prod}
\end{equation}
where $\overline{g(x)}$ is the quaternionic conjugate of $g(x)$, and (left)  scalar multiplication $af, \; a\in H,$ with $(af)(q) = af(q)$ (see \cite{Vis} for details). Similarly, one could define a right quaternionic Hilbert space of square integrable functions.

\subsection{Dual spaces}
In order to obtain a  quaternionic version of the Riesz representation theorem, we need to recall a few facts about the dual space of a quaternionic Hilbert space. We follow \cite{Alek} in order to do this.
Let $\mathfrak{H}_{ld}$ be the left dual space of the left quaternionic Hilbert space $V_H^L$. That is,
$$\mathfrak{H}_{ld}=\{h : V_H^L\longrightarrow H~ \mid ~h~ \text{is bounded and left linear}\}$$
with the usual norm $\|h\|= \sup\{\vert h(x)\vert \;\mid \;\|x\|=1,~x\in V_H^L\}$. It is known that $\mathfrak{H}_{ld}$ is a real vector space. Moreover, $\mathfrak{H}_{ld}$ can be transformed into a quaternionic Hilbert space. Indeed, as noted above, $V_H^L$ also admits a right multiplication by quaternionic scalars. Using this fact, for any functional $h\in\mathfrak{H}_{ld}$ and any $\lambda\in H$ define
$$(\lambda h)(x)=h(x\lambda),\quad (h\lambda)(x)=h(x)\lambda;~~~~~x\in V_H^L.$$
Then $\mathfrak{H}_{ld}$ becomes a two-sided quaternionic Banach space, with the scalar multiplication so defined.
\begin{theorem}[Riesz representation theorem] For any functional $h\in\mathfrak{H}_{ld}$, exactly as in the real and complex cases,
\begin{equation}\label{E13}
h(x)=\langle x\vert y\rangle,\quad x\in V_H^L
\end{equation}
for a vector $y\in V_H^L$, and then $\|h\|=\|y\|$.
\end{theorem}
Let $\{e_{\nu}\mid \nu\in\Lambda\}$ be a fixed orthonormal basis of $V_H^L$ and
define $J:\mathfrak{H}_{ld}\longrightarrow V_H^L$ by the cannonical mapping defined by the relation (\ref{E13}), that is by, $h(x)=\langle x\vert Jh\rangle;~~x\in V_H^L$. Then it can be shown that $J$ is additive, isometric and bijective. Now define
$$K:V_H^L\longrightarrow V_H^L, \qquad K(x)=\sum_{\nu\in\Lambda}\overline{\langle x\vert e_{\nu}\rangle}e_{\nu}.$$
Clearly, $K$ is additive. (Note,  $x = \sum_{\nu\in\Lambda}\langle x\vert e_{\nu}\rangle e_{\nu}$).
\begin{theorem}
The left dual space $\mathfrak{H}_{ld}$ of $V_H^L$, is also a two-sided quaternionic Hilbert space, if we introduce the inner product in $\mathfrak{H}_{ld}$ by
\begin{equation}\label{E14}
\langle h\vert k\rangle=\langle KJh\vert KJk\rangle;~~~~~h,k\in\mathfrak{H}_{ld}.
\end{equation}
The inner product (\ref{E14}) is consistent with the norm of $\mathfrak{H}_{ld}$.
\end{theorem}
Further, if we define $U:\mathfrak{H}_{ld}\longrightarrow V_H^L$ by $U=KJ$, then $U$ is a two-linear  bijective map, and thereby $\mathfrak{H}_{ld}$ is isomorphic to $V_H^L$.

\subsection{Cullen regular functions}
There have been a number of different suggestions in the literature on how the notion of holomorphy could be extended to functions of a quaternionic variable. We mention two here, of which we shall adopt the second definition. For a brief history of quaternionic holomorphy we refer the reader to \cite{Am}. The first definition of quaternionic holomorpy is given via the Cauchy-Fueter equations, which attempts to mimic the Cauchy-Riemann equations in a straightforward way.

\begin{definition}(Cauchy-Fueter equations) \cite{Am, Gra1}\label{def-cauchy-fueter}
Let $f:H\longrightarrow H$ be a quaternion valued function of a quaternionic variable. We say that $f$ is left-regular if it satisfies the Cauchy-Fueter equation
$$\frac{\partial_l f}{\partial{\overline{\qu}}}=\frac{\partial f}{\partial x_0}+ i\frac{\partial f}{\partial x_1}+ j\frac{\partial f}{\partial x_2}+ k\frac{\partial f}{\partial x_3}=0 ,$$
and that $f$ is right-regular if it satisfies the other Cauchy-Fueter equation
$$\frac{\partial_r f}{\partial{\overline{\qu}}}=\frac{\partial f}{\partial x_0}+\frac{\partial f}{\partial x_1}i+\frac{\partial f}{\partial x_2}j+\frac{\partial f}{\partial x_3}k=0 .$$
\end{definition}
Using this definition a theory of regular functions has been developed as a well structured theory.
However, the unpleasant feature of this definition is that under the Cauchy-Fueter equations the
function $f(\qu)=\qu$ is not regular, and thereby none of the monomials or polynomials are regular.
There have been several attempts to overcome this feature. The most promising, and recent attempt
has appeared in \cite{Gra1} (see also \cite{Gra2}) where the following definition is offered.
Let $\mathbb{S}=\{\qu=x_1i+x_2j+x_3j~\vert~x_1,x_2,x_3\in\mathbb{R},~x_1^2+x_2^2+x_3^2=1\}$.
\begin{definition}(Slice-regular functions \cite{Gra1})\label{D2}
Let $\Omega$ be a domain in $H$. A real differentiable (i.e., with respect to $x_0$ and the $x_i,\; i=1,2,3$) function $f:\Omega\longrightarrow H$ is said to be slice left regular if, for every quaternion $I\in\mathbb{S}$, the restriction of $f$ to the complex line $L_I=\mathbb{R}+I\mathbb{R}$ passing through the origin, and containing $1$ and $I$, has continuous partial derivatives (with
 respect to $x$ and $y$, every element in $L_I$ being uniquely expressible as $x + yI$) and satisfies
\begin{equation}
\overline{\partial}_I f (x + yI) := \frac 12\left(\frac {\partial f_I (x + yI )}{\partial x}
       + I \frac {\partial f_I (x + yI )}{\partial y}\right) = 0\; .
\label{leftslicereg}
\end{equation}
Similarly,  it is said to be slice right regular if
\begin{equation}
\overline{\partial}_I f (x + yI) := \frac 12\left(\frac {\partial f_I (x + yI )}{\partial x}
       +  \frac {\partial f_I (x + yI )}{\partial y}I\right) = 0\; .
\label{rightslicereg}
\end{equation}
\end{definition}

 With this definition all monomials of the form $a\qu^n,~a\in H,~n\in\mathbb{N}$, are slice right regular while those of the form $\qu^n a,\; a\in H,~n\in\mathbb{N}$, are slice left regular . Since regularity respects addition, all polynomials of the form $f(\qu)=\sum_{t=0}^{n}a_t\qu^t$, with $a_t\in H$, are slice right regular and similarly polynomials of the form $f(\qu)=\sum_{t=0}^{n}\qu^t a_t$, are slice left regular. Further, an analog of Abel's theorem guarantees convergence of appropriate infinite power series.
\begin{proposition}\cite{Gra1}\label{P1}
For any non-real quaternion $\qu\in H-\mathbb{R}$, there exist, and are unique, $x,y\in\mathbb{R}$ with $y>0$, and $I\in\mathbb{S}$ such that $\qu=x+yI$.
\end{proposition}
Henceforth we shall refer to slice right (left) regular functions simply as right (left) regular functions. We now define the Cullen derivative of regular functions.
\begin{definition}(Cullen derivative \cite{Gra1,Gra2})\label{D3}
Let $\Omega$ be a domain in $H$, and let $f:\Omega\longrightarrow H$ be a left regular function. The Cullen derivative of $f$, $\partial_cf$, is defined as
$$\partial_c   f (\qu)=\left\{\begin{array}{ccc}
\partial_If(\qu):=\dfrac{1}{2}\left(\dfrac{\partial f_I(x+Iy)}{\partial x}-I\dfrac{\partial f_I(x+Iy)}{\partial y}\right)&\text{if}&y\not=0\\[4mm]
\dfrac{\partial f}{\partial x}(x)&\text{if}&\qu =x~~\text{is real}\end{array}
\right.$$
Similarly, for a right regular function $f$ its Cullen derivative is defined as
$$\partial_c   f (\qu)=\left\{\begin{array}{ccc}
\partial_If(\qu):=\dfrac{1}{2}\left(\dfrac{\partial f_I(x+Iy)}{\partial x}-\dfrac{\partial f_I(x+Iy)}{\partial y} I\right)&\text{if}&y\not=0\\[4mm]
\dfrac{\partial f}{\partial x}(x)&\text{if}&\qu =x~~\text{is real}\end{array}
\right.$$
\end{definition}
Under the above definition the Cullen derivative of a regular function is regular and with $a_n\in H$ we have, for example, for a right regular power series,
\begin{equation}\label{E11}
\partial_c\left(\sum_{n=0}^{\infty}a_n \qu^n\right)=\sum_{n=0}^{\infty} n a_n \qu^{n-1}.
\end{equation}
The following theorem gives the quaternionic version of holomorphy via a Taylor series. Let $B(0,R)$ be an open ball in $H$, of radius $R$ and centered at $0$.
\begin{theorem}\cite{Gra1}\label{T5}
A function $f: B(0,R)\longrightarrow H$ is right, respectively left, regular if and only if it has a series expansion of the form
$$ f(\qu)=\sum_{n=0}^{\infty}\frac{1}{n!}\frac{\partial^n f}{\partial x^n}(0)\qu^n , \qquad\text{respectively,} \qquad f(\qu)=\sum_{n=0}^{\infty}\qu^n\frac{1}{n!}\frac{\partial^n f}{\partial x^n}(0),$$
converging on $B(0, R)$.
\end{theorem}

\section{Complex Hermite polynomials}
For the construction and  analysis of the quaternionic Hermite polynomials, $H_n(\qu)$ and $H_{n,m}(\qu,\overline{\qu})$, it would be useful to first review some facts about their complex counterparts, specially since the results we obtain here in the  quaternionic case are parallel to those in the complex case. For a detailed discussion of the complex Hermite polynomials $H_n(z)$ and $H_{n,m}(z,\overline{z})$, we refer the reader to \cite{AFG,Ga,Ghan,Mat,Wil,VE}.
\subsection{The polynomials $H_n(z)$}
Let $z=x+iy\in\C, \;\; 0<s<1$,
\begin{equation}
d\nu(z)=d\nu(x,y)=\exp\left[-(1-s)x^2-(\frac{1}{s}-1)y^2\right]dxdy
\label{measure1}
\end{equation}
and
\begin{equation}
b_n(s)=\frac{\pi\sqrt{s}}{1-s}\left(2\frac{1+s}{1-s}\right)^n n!\, .
\label{normalization}
\end{equation}
Define
\begin{equation}\label{H1}
H_n(z)=n!\sum_{m=0}^{[n/2]}\frac{(-1)^m(2z)^{n-2m}}{m!(n-2m)!}=n!\sum_{m=0}^{[n/2]}C_{nm}z^{n-2m}.
\end{equation}
($[x]$ denoting the integer part of $x$)
and observe that $\overline{H_n(x+iy)}=H_n(x-iy)$. The $H_n$ are just the usual Hermite polynomials, written in terms of a complex argument. From \cite{Ga} and \cite{VE} we have
\begin{equation}\label{H2}
\int_{\R^2}H_n(x+iy)\overline{H_m(x+iy)}d\nu(x,y)=b_n(s)\; \delta_{nm}\;  .
\end{equation}
Define
\begin{equation}\label{F3}
\mathfrak{h}_{n,s}(z)= b_n(s)^{-\frac{1}{2}}e^{-z^2/2}H_n(z)
\end{equation}
and the Hilbert space of entire functions
$$\mathfrak{X}_s=\left\{f~\mid~\int_{\R^2}|f(x+iy)|^2\exp(sx^2-\frac{1}{s}y^2)dxdy<\infty\right\}.$$
It has been shown in \cite{VE} that $\{\mathfrak{h}_{n,s}(z)\}_{n=0}^{\infty}$ is an orthonormal basis of $\mathfrak{X}_s$ and $\mathfrak{X}_s$ is a reproducing kernel Hilbert space with the kernel
\begin{equation}\label{F4}
\mathfrak{K}_s(z,\overline{w})=\sum_{n}\mathfrak{h}_{n,s}(z)\overline{\mathfrak{h}_{n,s}(w)}=\frac{1-s^2}{2\pi s}\exp\left[-\frac{1+s^2}{4s}(z^2+\overline{w}^2)+\frac{1-s^2}{2s}z\overline{w}\right],
\end{equation}
where $\quad z,w\in\C$. The expression for the kernel  can be reduced to (see \cite{Ga})
\begin{equation}\label{F5}
\mathfrak{K}_s(z,\overline{z})=\frac{1-s^2}{2\pi s}\exp[-sx^2+\frac{1}{s}y^2],\quad z=x+iy.
\end{equation}
If we take
\begin{equation}\label{H3}
h_{n,s}(z)= b_n(s)^{-\frac{1}{2}}H_n(z)
\end{equation}
then for $z,w\in\mathbb{C}$,
\begin{equation}\label{H4}
K_s(z,\overline{w})=\sum_{n}h_{n,s}(z)\overline{h_{n,s}(w)}=\frac{1-s^2}{2\pi s}\exp\left[-\frac{(s-1)^2}{4s}(z^2+\overline{w}^2)+\frac{1-s^2}{2s}z\overline{w}\right].
\end{equation}
Further by replacing $w$ by $z$ we get
\begin{equation}\label{H5}
K_s(z,\overline{z})=\frac{1-s^2}{2\pi s}\exp\left[\frac{1-s}{2}x^2+\frac{s^2-3s+2}{2s}y^2\right],\quad z=x+iy.
\end{equation}
The kernel (\ref{H4}) is also a reproducing kernel. In fact it is easily seen that the corresponding
reproducing kernel Hilbert space, $\mathfrak H^s_{hol}$, which is again a Hilbert space of
analytic functions, and for which the vectors (\ref{H3}) form
an orthonormal basis, is a subspace of the Hilbert space $L^2 (\mathbb C , d\mu_s )$, where
\begin{equation}
  d\mu_s (x,y) = e^{-[(1-s)x^2 + (\frac 1s -1)y^2]}\; dx\; dy\; .
\label{s-measure}
\end{equation}
Similarly, the polynomials $h_{n,s} (\overline{z})$ would span a reproducing kernel Hilbert space
 $\mathfrak H_{a-hol} \subset L^2 (\mathbb C , d\mu_s )$ of
anti-analytic functions, with reproducing kernel
$\overline{K_s(z,\overline{w})}$.

\subsection{The polynomials $H_{n,m}(z,\overline{z})$}\label{subsec-comp-herm}
A second class of complex Hermite polynomials have been studied in \cite{Wil}. These are defined, for positive integers $n,m$, as
\begin{equation}\label{E1}
H_{n,m}(z,\overline{z})=(-2)^{n+m}\text{exp}(z\overline{z}/2)\left(\frac{\partial}{\partial z}\right)^n\left(\frac{\partial}{\partial\overline{z}}\right)^m\text{exp}(-z\overline{z}/2).
\end{equation}
The generating function for this polynomial is
\begin{equation}\label{E2}
\text{exp}[(a\overline{z}+\overline{a}z-a\overline{a})/2]=\sum\frac{(a/2)^n(\overline{a}/2)^m}{n!m!}H_{n,m}(z,\overline{z}).
\end{equation}
When $m\geq n$ and $n\geq 0$, the complex Hermite polynomials can also be written as
$$H_{n,m}(z,\overline{z})=\frac{m!}{(m-n)!}z^{m-n}(-2)^n{}_1F_1(-n,m-n+1,z\overline{z}/2).$$
These polynomials satisfy the orthogonality relation
\begin{equation}\label{E3}
\int_{-\infty}^{\infty}\int_{-\infty}^{\infty}\text{exp}(-z\overline{z}/2)H_{n,m}(z,
\overline{z})\overline{H_{\nu,\mu}(z,\overline{z})}\; dx\; dy =2\pi \delta_{n \nu}\delta_{m \mu}n!m!2^{n+m}.
\end{equation}
Now let
\begin{equation}\label{E4}
h_{n,m}(z,\overline{z})=(-1)^{n+m}\text{exp}(|z|^2)\left(\frac{\partial}{\partial z}\right)^n
\left(\frac{\partial}{\partial\overline{z}}\right)^m\text{exp}(-|z|^2)=\sum_{i=0}^{n}
\sum_{j=0}^{m}a_{ij}z^i\oz^j.
\end{equation}
From (\ref{E3}) (see also \cite{Nic,Ghan} ), the polynomials $h_{n,m}(z,\overline{z})$
form a
complete orthogonal system in the Hilbert space $L^2(\mathbb{C}, e^{-|z|^2}d\lambda)$, where
$d\lambda=\frac{1}{\pi}d^2z$ is the Lebesque measure on $\mathbb{C}$. For proof of completeness
see \cite{Ghan}. Also from (\ref{E3}) we have
\begin{equation}\label{E5}
\int_{\mathbb{C}}e^{-|z|^2}h_{n,m}(z,\overline{z})\overline{h_{n,m}(z,\overline{z})}\; d\lambda
=n!m!
\end{equation}
and thereby
$$\|h_{n,m}\|_{L^2}=\sqrt{n!m!}.$$
 The operator
\begin{equation}\label{E6}
\mathfrak{L}=-\frac{1}{4}\left\{4\frac{\partial^2}{\partial z\partial\oz}+
2\left(z\frac{\partial}{\partial z}-\oz\frac{\partial}{\partial\oz}\right)-\abs{z}^2\right\}
\end{equation}
 is the Hamiltonian of a nonrelativistic quantum particle moving on the plane under
 the action of a constant external magnetic field applied perpendicularly to the plane.
 The functions $e^{-\abs{z}^2/2}h_{m,n}(z,\oz)$ are the eigenfunctions of $\mathfrak{L}$ with
 eigenvalues $n+\frac{1}{2}$, each eigenvalue being infinitely degenerate. The degeneracy is
 given by  $m=0,1,2,\cdots$. For details see  \cite{Ghan,Mat}.\\

\section{Quaternionic Hermite polynomials}\label{sec-quherm-polyn}
In this section we define the quaternionic Hermite polynomials $H_n(\qu)$ and $H_{n,m}(\qu,\overline{\qu})$, by analogy with the complex polynomials introduced above.  We shall identify the set of all polynomials $H_n(\qu)$ as an orthogonal set in a $L^2$-space and similarly for the $H_{n,m}(\qu,\overline{\qu})$. We also obtain reproducing kernels using the polynomials $H_n(\qu)$ and $H_{n,m}(\qu,\overline{\qu})$ and the corresponding reproducing kernel Hilbert spaces.  Finally we look at an operator $\mathfrak{L}_H$ which can be considered as the quaternionic version of the Landau operator (\ref{E6}).

 It is well-known (see, e.g.,  \cite{Thi}) that any $\qu\in H$, in the $2\times 2$ matrix representation, can be written as
\begin{equation}\label{dc}
\qu=u_{\qu}Zu_{\qu}^{\dagger},
\end{equation} where
$$u_{\qu}=\left(\begin{array}{cc}
e^{i\psi/2}&0\\
0&e^{-i\psi/2}\end{array}\right)
\left(\begin{array}{cc}
\cos\frac{\phi}{2}&i\sin\frac{\phi}{2}\\
i\sin\frac{\phi}{2}&\cos\frac{\phi}{2}\end{array}\right)
\left(\begin{array}{cc}
e^{i\psi/2}&0\\
0&e^{-i\psi/2}\end{array}\right)\in SU(2),$$
and $Z=\left(\begin{array}{cc}
z&0\\
0&\oz\end{array}\right), \; z \in \mathbb C$. Since $SU(2)$ is a compact group, let $d\omega(u_{\qu})$ be the normalized Haar measure on it.
From the decomposition (\ref{dc}), $d\mu (\qu) := e^{-|z|^2}d\lambda d\omega(u_{\qu})$ is a measure on $H$ and $L^2_H(H, e^{-|z|^2}d\lambda d\omega(u_{\qu}))$ can be considered to be a {\em left quaternionic Hilbert space}, with an inner product defined as in (\ref{left-sc-prod}).

From (\ref{dc}), for any $i,j\in\mathbb{N}$ we have
$$\qu^i=u_{\qu}\left(\begin{array}{cc}
z^i&0\\
0&\oz^i\end{array}\right)u_{\qu}^{\dagger}\quad\text{and}\quad
{\overline{\qu}}^j=u_{\qu}\left(\begin{array}{cc}
\oz^j&0\\
0&z^j\end{array}\right)u_{\qu}^{\dagger}$$
and thereby,
\begin{equation}\label{E7}
\qu^i\overline{\qu}^j=u_{\qu}\left(\begin{array}{cc}
z^i\oz^j&0\\
0&\oz^iz^j\end{array}\right)u_{\qu}^{\dagger}.
\end{equation}

\subsection{The quaternionic Hermite polynomials $H_n(\qu)$}
 Let $d\eta(\qu)=d\mu_s(z)d\omega(u_{\qu})$, with $d\mu_s$ as in (\ref{s-measure}). Since
$$\qu^{n-2m}=u_{\qu}\left(\begin{array}{cc}z^{n-2m}&0\\0&\overline{z}^{n-2m}\end{array}\right)
u_{\qu}^{\dagger}$$
and the $C_{nm}$ in (\ref{H1}) are real numbers, we have
$$C_{nm}\qu^{n-2m}=u_{\qu}\left(\begin{array}{cc}C_{nm}z^{n-2m}&0\\0&C_{nm}\overline{z}^{(n-2m)}
\end{array}\right)u_{\qu}^{\dagger}$$
and thereby
$$n!\sum_{m=0}^{[n/2]}C_{nm}\qu^{n-2m}=u_{\qu}\left(\begin{array}{cc}n!\sum_{m=0}^{[n/2]}
C_{nm}z^{n-2m}&0\\0&n!\sum_{m=0}^{[n/2]}C_{nm}\overline{z}^{(n-2m)}\end{array}\right)
u_{\qu}^{\dagger}$$
That is
$$H_n(\qu)=u_{\qu}\left(\begin{array}{cc}H_n(z)&0\\0&H_n(\overline{z})\end{array}\right)
u_{\qu}^{\dagger}.$$
Observe that
$$\overline{H_n(\qu)}=H_n(\overline{\qu})=u_{\qu}\left(\begin{array}{cc}H_n(\overline{z})&0\\
0&H_n(z)\end{array}\right)u_{\qu}^{\dagger}.$$

Similarly, from (\ref{H1}) we easily see that
\begin{equation}
  H(\qu ) = n!\sum_{m=0}^{[n/2]}\frac{(-1)^m(2\qu )^{n-2m}}{m!(n-2m)!}\;,
\label{q-herm-exp}
\end{equation}
which means that they satisfy the same recursion relations,
\begin{equation}
   \qu H_n (\qu ) = \frac 12 H_{n+1}(\qu ) + n H_{n-1} (\qu )\; ,
\label{q-herm-rec}
\end{equation}
as the real Hermite polynomials $H_n (x)$, or also
\begin{equation}
H_{n+1} (\qu ) = 2\qu H_n(\qu ) - H^\prime_n(\qu )\; ,
\label{q-herm-rec2}
\end{equation}
the prime denoting the (Cullen) derivative.

Let
$$L_H^2(H,d\eta(\qu))=\left\{f:H\rightarrow H\;\; \left|\;\;\int_Hf(\qu)\overline{f(\qu)}d\eta(\qu)<\infty \right.\right\}.$$
This is a left quaternionic Hilbert space with a scalar product as in (\ref{left-sc-prod}).
\begin{theorem}\label{TH1}
The set $\left\{H_n(\qu)\;\;\left| \;\; \right.n\in\mathbb{N}\right\}$ is an orthogonal set in $L_H^2(H,d\eta(\qu))$.
\end{theorem}
\begin{proof}
Consider
\begin{eqnarray*}
& &\int_HH_n(\qu)\overline{H_m(\qu)}d\eta(\qu)\\
&=&\int_{SU(2)}u_{\qu}\left(\begin{array}{cc}
\int_{\R^2}H_n(z)H_m(\overline{z})d\nu(x,y)&0\\0&\int_{\R^2}H_m(z)H_n(\overline{z})d\nu(x,y)\end{array}\right)u_{\qu}^{\dagger}d\omega(u_{\qu})\\
&=&\int_{SU(2)}u_{\qu}\left(\begin{array}{cc}
b_n(s)n!\delta_{mn}&0\\
0&b_m(s)m!\delta_{mn}\end{array}\right)u_{\qu}^{\dagger}d\omega(u_{\qu})\\
&=&\int_{SU(2)}u_{\qu}u_{\qu}^{\dagger}d\omega(u_{\qu})b_n(s)n!\delta_{mn}\\
&=&b_n(s)n!\delta_{nm}.
\end{eqnarray*}
\end{proof}
Redefine
\begin{equation}\label{C3}
H_n^s (\qu )=b_n(s)^{-\frac{1}{2}}H_n(\qu),
\end{equation}
then $H_n^s(\qu)\in L_H^2(H,d\eta(\qu))$ and
\begin{equation}\label{H6}
\int_HH_n^s(\qu)\overline{H_m^s(\qu)}d\eta(\qu)=\delta_{nm}.
\end{equation}
Define the kernel
\begin{equation}\label{H7}
K_s(\qu_1,\overline{\qu_2})=\sum_{n=0}^{\infty}H_n^s(\qu_1)\overline{H_n^s(\qu_2)},
\end{equation}
then from (\ref{H4}) $K_s(\qu_1,\qu_2)$ is a reproducing kernel and the corresponding reproducing kernel Hilbert space is
\begin{equation}\label{C1}
A_s=\overline{\text{span}\{H_n^s(\qu) \mid n\in\mathbb{N}\}},
\end{equation}
the span and its closure being taken under left multiplication by quaternionic constants.
Note that the above reproducing kernel is the quaternionic equivalent of the kernel (\ref{H4}) and the
Hilbert space $A_s$ the equivalent of $\mathfrak H_{hol}^s$, generated by that kernel.
\begin{lemma}\label{L1}
$K_s(\qu,\overline{\qu})=K_s(z,\overline{z})\mathbb{I}_2.$
\end{lemma}
\begin{proof}
Since $$H_n(\qu)=u_{\qu}\left(\begin{array}{cc}H_n(z)&0\\0&H_n(\overline{z})\end{array}\right)
u_{\qu}^{\dagger},$$
we have
$$H_n(\qu)\overline{H_n(\qu)}=u_{\qu}\left(\begin{array}{cc}H_n(z)\overline{H_n(z)}&0\\0&H_n(z)\overline{H_n(z)}\end{array}\right)u_{\qu}^{\dagger}.$$
Since $b_n(s)$ are real numbers,
\begin{eqnarray*}
b_n(s)^{-1}H_n(\qu)\overline{H_n(\qu)}&=&u_{\qu}\left(\begin{array}{cc}b_n(s)^{-1}H_n(z)\overline{H_n(z)}&0\\
0&b_n(s)^{-1}H_n(z)\overline{H_n(z)}\end{array}\right)u_{\qu}^{\dagger}\\
&=&u_{\qu}\left(\begin{array}{cc}h_{n,s}(z)\overline{h_{n,s}(z)}&0\\0&h_{n,s}(z)\overline{h_{n,s}(z)}\end{array}\right)u_{\qu}^{\dagger}.
\end{eqnarray*}
Therefore
\begin{eqnarray*}
\sum_{n}b_n(s)^{-1}H_n(\qu)\overline{H_n(\qu)}&=&u_{\qu}\left(\begin{array}{cc}\sum_nh_{n,s}(z)\overline{h_{n,s}(z)}&0\\0&\sum_nh_{n,s}(z)\overline{h_{n,s}(z)}\end{array}\right)u_{\qu}^{\dagger}\\
&=&u(\qu)\left(\begin{array}{cc}K_s(z,z)&0\\0&K_s(z,z)\end{array}\right)u_{\qu}^{\dagger}\\
&=&K_s(z,\overline{z})\mathbb{I}_2.
\end{eqnarray*}
That is,
\begin{equation}\label{C2}
K_s(\qu,\overline{\qu})=K_s(z,\overline{z})\mathbb{I}_2.
\end{equation}
\end{proof}

\subsection{The quaternionic Hermite polynomials $H_{n,m}(\qu,\overline{\qu})$}
The exponential series
$$e^{\qu}=\sum_{n=0}^{\infty}\frac{\qu^n}{n!},\quad \qu\in H,$$
converges absolutely, and uniformly on compact subsets with respect to the norm of $H$ \cite{Ebb} (p 204). Thereby
$$e^{-|\qu|^2}=e^{-\qu\overline{\qu}}=\sum_{n=0}^{\infty}(-1)^n\frac{(\qu\overline{\qu})^n}{n!}$$
converges uniformly. Further $\qu$ and $\overline{\qu}$ commute and real numbers commute with quaternions. Therefore, as in the complex case, as an extension of complex hermite polynomials, using Definition \ref{D3} we get
\begin{eqnarray}\label{E8}
h_{n,m}(\qu,\overline{\qu})&=&(-1)^{n+m}e^{|\qu|^2}\frac{\partial^{n+m}}
{\partial\qu^n\partial\overline{\qu}^m}e^{-|\qu|^2}=\sum_{i=0}^{n}\sum_{j=0}^{m}a_{ij}\qu^i\overline{\qu}^j\\
 &=& n!m!\sum_{j=0}^{\text{min}\{n,m\}}\frac{(\overline{\qu})^{n-j}}{(n-j)!}\frac{\qu^{m-j}}{(m-j)!},
\end{eqnarray}
where the $a_{ij}$ are real coefficients and the derivatives should be understood in the Cullen sense. Further from (\ref{E7}) we have
\begin{equation*}
a_{ij}\qu^i\overline{\qu}^j=u_{\qu}\left(\begin{array}{cc}
a_{ij}z^i\oz^j&0\\
0&a_{ij}\oz^iz^j\end{array}\right)u_{\qu}^{\dagger}.
\end{equation*}
Therefore
\begin{equation*}
\sum_{i=0}^{n}\sum_{j=0}^{m}a_{ij}\qu^i\overline{\qu}^j=u_{\qu}\left(\begin{array}{cc}
\sum_{i=0}^{n}\sum_{j=0}^{m}a_{ij}z^i\oz^j&0\\
0&\sum_{i=0}^{n}\sum_{j=0}^{m}a_{ij}\oz^iz^j\end{array}\right)u_{\qu}^{\dagger}.
\end{equation*}
Thereby from (\ref{E8}) we get
\begin{equation}\label{E9}
h_{n,m}(\qu,\overline{\qu})=u_{\qu}\left(\begin{array}{cc}
h_{n,m}(z,\oz)&0\\
0&h_{m,n}(z,\oz)\end{array}\right)u_{\qu}^{\dagger}.
\end{equation}
We can see from (\ref{E4}) that $h_{n,m}(\qu,\overline{\qu})$ is obtained by substituting $z$ by $\qu$ in the definition (\ref{E4}) of $h_{n,m}(z,\oz)$. Thereby, using the fact that $h_{n,m}(z,\oz)\in L^2(\mathbb{C}, e^{-\abs{z}^2}d\lambda)$ and $SU(2)$ is a compact group, we have $h_{n,m}(\qu,\overline{\qu})\in L^2_H(H, e^{-\abs{z}^2}d\lambda d\omega(u_{\qu}))$.
\begin{theorem}\label{T1}
The set $\left\{h_{m,n}(\qu,\overline{\qu})~\vert~m,n\in\mathbb{N}\right\}$ is an orthogonal set in $L_H^2(H, e^{-\abs{z}^2}$ $d\lambda d\omega(u_{\qu})).$
\end{theorem}
\begin{proof}
See \cite{Thi2}.
\end{proof}
\begin{theorem}\label{T2}
$\|h_{m,n}(\qu,\overline{\qu})\|_{L_H^2}=\sqrt{n!m!}$ in $L_H^2(H, e^{-\abs{z}^2}d\lambda d\omega(u_{\qu})).$
\end{theorem}
\begin{proof} From (\ref{E5}) and (\ref{E9}) we have
\begin{eqnarray*}
&&\int_{H}h_{n,m}(\qu,\overline{\qu})\overline{h_{n,m}(\qu,\overline{\qu})}e^{-\abs{z}^2}d\lambda d\omega(u_{\qu})\\
&=&\int_{SU(2)}\int_{\mathbb{C}}u_{\qu}\left(\begin{array}{cc}
h_{n,m}(z,\oz)\overline{h_{n,m}(z,\oz)}&0\\
0&h_{m,n}(z,\oz)\overline{h_{m,n}(z,\oz)}\end{array}\right)u_{\qu}^{\dagger}e^{-\abs{z}^2}d\lambda d\omega(u_{\qu})\\
&=&\int_{SU(2)}u_{\qu}\left(\begin{array}{cc}
\int_{\mathbb{C}}h_{n,m}(z,\oz)\overline{h_{n,m}(z,\oz)}e^{-\abs{z}^2}d\lambda&0\\
0&\int_{\mathbb{C}}h_{m,n}(z,\oz)\overline{h_{m,n}(z,\oz)}e^{-\abs{z}^2}d\lambda\end{array}\right)\\
&\times & u_{\qu}^{\dagger}
d\omega(u_{\qu})\\
&=&n!m!\int_{SU(2)}u_{\qu}u_{\qu}^{\dagger}
d\omega(u_{\qu})\\
&=&n!m!\mathbb{I}_2.
\end{eqnarray*}
\end{proof}

 From \cite{Ghan,Mat}, (\ref{E6}) and (\ref{E9}) it is clear that the functions
 $e^{-\abs{\qu}^2/2}h_{n,m}(\qu,\overline{\qu})$ are the eigenfunctions of the operator
$$\mathfrak{L}_{H}=u_{\qu}\left(\begin{array}{cc}\mathfrak{L}&0\\0&\overline{\mathfrak{L}}
\end{array}\right)u_{\qu}^{\dagger}$$
with spectrum $n+\frac{1}{2}$, each level being infinitely degenerate ($m=0,1,2,3,...\;$).
Even though this operator can be considered as the quaternionic version of the complex
Landau operator, we do not have at the moment a proper physical understanding of it.

\section{Coherent states}\label{sec-CS}
In this section we define CS over quaternionic Hilbert spaces and in particular CS arising from quaternionic Hilbert spaces of (slice) regular functions. As examples we build CS using the quaternionic Hermite polynomials $H_n(\qu)$ and $H_{n,m}(\qu,\overline{\qu})$.

\subsection{The general construction}\label{subsec-gen-cons}
Coherent states may be built on quaternionic Hilbert spaces, in more or less the same way as
was outlined in Section \ref{sec_intro}, for coherent states on complex Hilbert spaces. Indeed,
let $V^L_H$ be a left quaternionic Hilbert space whose dimension could be finite or countably
infinite and let $\phi_m , \;\; m =0,1,2, \ldots$, be an orthonormal basis of this space.
Let $X$ be a locally compact space and $\mu$ a (Radon) measure on it. Consider a set of functions
$\Phi_m : X \longrightarrow H , \;\; m =0,1,2, \ldots ,$ of the same cardinality as the dimension
of $V^L_H$, and which satisfy the two conditions,
\begin{itemize}
\item[1.] $ 0 < \mathcal N (x) := \sum_m\vert\Phi_m (x)\vert^2 < \infty$, for all $x \in X$.

\medskip
\item[2.] $\int_X \Phi_m (x)\overline{\Phi_n (x)}\; d\mu (x) = \delta_{m n}$, for all $m$ and $n$.
    \end{itemize}

A family of coherent states $\{ \eta_x \mid x \in X\}\subset V^L_H$ can now be defined to be the vectors,
\begin{equation}
\eta_x = \mathcal N (x)^{- \frac 12} \sum_m \Phi_m (x) \phi_m\; .
\label{quat-CS}
\end{equation}

  By construction, these coherent states are seen to be normalized, i.e., $\Vert \eta_x \Vert^2 =1$, for all $x\in X$, and to satisfy the resolution of the identity,
$$ \int_X \langle f \mid \eta_x \rangle \langle \eta_x \mid g \rangle \; \mathcal N (x)d\mu (x)\;
= \langle f \mid g\rangle \; , \qquad f,g \in V^L_H \; .$$
Moreover, taking $L^2_H (X, d\mu )$ to be a left quaternionic Hilbert space, the map
\begin{equation}
W: V^L_H \longrightarrow L^2_H (X, d\mu ), \quad \text{with} \quad  W f (x) = \mathcal N (x)^{\frac 12}\langle f\mid \eta_x\rangle_{V^L_H}
\label{CS-isom}
\end{equation}
 is a linear isometry onto a closed subspace
 $$ \mathfrak H_K := WV^L_H \subset  L^2_H (X,  d\mu ).$$
 The subspace $\mathfrak H_K$ is a {\em reproducing kernel Hilbert space}, with reproducing kernel
\begin{equation}
  K: X \times X \longrightarrow H, \qquad K(y,x ) = \left[\mathcal N (y)\; \mathcal N (x)\right]^{\frac 12}\; \langle \eta_y \mid \eta_x\rangle = \sum_m \Phi_m (y)\overline{\Phi_m (x)}\; .
\label{quat-repker}
\end{equation}
Thus, if $F \in \mathfrak H_K$, so that $F(x) = \mathcal N(x)^{\frac 12}\langle f \mid \eta_x\rangle$, for some $f\in V^L_K$, then
$$
 \int_X F(y) K(x, y)\; d\mu (y) = F(x),$$
for all $x \in X$. This also means that for each $x\in X$, the {\em evaluation map}, $E_x : \mathfrak H_K \longrightarrow H$, with $E_x (F) = F(x)$, is continuous and
\begin{equation}
\vert F(x)\vert \leq \mathcal N(x)^{\frac 12}\Vert f\Vert_{V^L_H} = \mathcal N(x)^{\frac 12}\Vert F\Vert_{\mathfrak H_K}.
\label{eval-est}
\end{equation}

All these results are familiar from the theory of coherent states on complex Hilbert spaces. Thus, entirely analogous results hold on quaternionic Hilbert spaces.

  Suppose, in particular, that $X = \mathcal D$ (some open ball in $H$, centered at the origin) and that the functions $\Phi_m, \;\; m =0,1,2, \ldots ,$ are elements in $L^2_H (\mathcal D, d\mu )$, which are regular functions, whose Taylor expansions (around the origin) have {\em real coefficients} (for example, they could be normalized polynomials, of degree $m$ in the quaternionic variable $\qu$, with real coefficients). Then, for any $f \in V^L_H$ the transformed function $Wf \in L^2_H (\mathcal D, d\mu )$ is a series in the conjugate variable $\overline{\qu}$  i.e., it  has the form
  $$ Wf (\overline{\qu }) = \sum_m \alpha_m \overline{\qu}^m , \qquad \alpha_m \in H\; , $$
the sum converging in the $L^2$-norm.  We now show that the above series is in fact a {\em right  anti-regular} function (i.e., regular in the variable  $\overline{\qu}$). For this we need the following lemma, which is an adaptation of Lemma 2.5 in \cite{Gra1}.
\begin{lemma}\label{split-lemma}
If $f$ is a right regular function on $B = B(0, R)$, then for every $I \in \mathbb S$ and every $J$ in $\mathbb S$, perpendicular to $I$, there exist two holomorphic functions $F, G : B\cap L_I \longrightarrow L_I$ such that the restriction of $f$ to $L_I$ can be split as the sum
\begin{equation}
  f_I (z ) = F(z) + J G(z), \qquad z = x + y I \in L_I\; .
\label{splitting}
\end{equation}
\end{lemma}
Note that if $f$ were to be an anti-regular function, a similar splitting into two anti-holomorphic functions would hold.
\begin{theorem}\label{reg-hilb-sp}
The  reproducing kernel Hilbert space $\mathfrak H_K$  is a space of right anti-regular functions.
\end{theorem}

\begin{proof}
 First note that in view of the boundedness condition
 $$
  \mathcal N (\qu) = \sum_m \vert \Phi_m (\qu)\vert^2 < \infty, \qquad \qu \in \mathcal D,$$
to each $\qu \in \mathcal D$ and  $I \in \mathbb S$, there exists a neighbourhood of the
origin, $N(\qu) \subset L_I \cap \mathcal D$, in which the above sum converges uniformly.
 Let $\phi_m$ be one  of the basis vectors in $V^L_H$, used to build the coherent
 states (\ref{quat-CS}). By (\ref{CS-isom}), $W\phi_m = \overline{\Phi}_m$ and since
 $\Phi_m$ has a series expansion with real coefficients, $\overline{\Phi_m (\qu)} =
 \Phi_m (\overline{\qu})$. Thus these vectors form a basis for $\mathfrak H_K$. A
 general element $f \in \mathfrak H_K$ is thus an $L^2$-limit  of a sequence of
 vectors $f_n , \;\; n =1,2,3, \ldots$, each one of which is a finite (left)
 quaternionic linear combination of these basis vectors. Hence each $f_n$ is a right
 anti-regular function. By (\ref{eval-est}) $\vert f_n(\overline{\qu}) -
 f(\overline{\qu})\vert \leq \vert \mathcal N (\qu)\vert\; \Vert f_n - f\Vert_{V^L_H}$
 and thus for each $\qu \in \mathcal D$, the sequence $f_n (\overline{\qu})$ converges
 to $f(\overline{\qu})$ uniformly in $\overline{N(\qu)}$. Since each $f_n$ is an
 anti-regular function, and since the uniform limit of an anti-holomorphic function is
 anti-holomorphic, Lemma \ref{split-lemma} implies that $f$ is also anti-regular, proving
 the theorem.
\end{proof}

Henceforth we shall denote the space $\mathfrak H_K$ by  $\mathfrak H_{a-reg}$. In
an entirely analogous manner we can also define a reproducing kernel Hilbert space
$\mathfrak H_{reg}$ of regular functions, starting from coherent states built out
of the polynomials $H_m (\overline{\qu})$ instead of $H_m (\qu)$. Both these spaces
are contained as subspaces of $L^2 (H, d\mu)$.

\subsection{Two examples}\label{subsec-examples}
We now construct coherent states following the above procedure and using the Hermite polynomials introduced in Section \ref{sec-quherm-polyn}. From (\ref{H4}) and Lemma \ref{L1} it is clear that the functions $H^s_n(\qu )$ in (\ref{C3}) satisfy the  conditions 1. and 2. stated at the beginning of Section \ref {subsec-gen-cons}. Thus, we have the result:

\begin{theorem}\label{reg_CS_theor}Let $\{\phi_n\}_{n=0}^{\infty}$ be an orthonormal 
basis of the left quaternionic Hilbert space $V_H^L$.
For $\qu\in H$, $0<s<1$, $K_s(\qu,\overline{\qu})$ as in (\ref{C2}) and $H_n^s(\qu)$ as in (\ref{C3}), 
the set of vectors
\begin{equation}\label{H8}
 \eta_{\qu,s}=K_s(\qu,\overline{\qu})^{-\frac{1}{2}}\sum_{n=0}^{\infty}H_n^s(\qu)\phi_n\in V_H^L
\end{equation}
forms a set of coherent states.
\end{theorem}
By Theorem \ref{reg_CS_theor}, the reproducing kernel Hilbert space associated to this
family of CS is a space of right anti-regular functions. Similarly, had we constructed these
CS using the functions $H^s_n (\overline{\qu})$, the corresponding reproducing kernel Hilbert
space would have consisted of right regular functions, and in fact would have been
the space $A_s$ in (\ref{C1}).

\bigskip

 Next, for each fixed $n$, let  $B_n=\left\{h_{n,m}(\qu,\overline{\qu})\mid m\in\mathbb{N},
 \right\}$ and let $A_n(H)$ be the closed linear span, under left multiplication by
 quaternionic scalars, of the vectors in $B_n$. Then $B_n$ is a basis of $A_n(H)$ and
$\bigoplus_{n=0}^{\infty}A_n(H)$ is a left quaternionic Hilbert space, which is a closed
subspace of $L^2_H(H, e^{-\abs{z}^2}d\lambda d\omega(u_{\qu})).$ Note that unlike in the case of
the complex polynomials, discussed in Section \ref{subsec-comp-herm}, where the vectors
$h_{n,m}$ spanned the entire space $L^2(\mathbb{C}, e^{-|z|^2}d\lambda)$, here
$\bigoplus_{n=0}^{\infty}A_n(H)$ is only a proper subspace of
 $L^2_H(H, e^{-\abs{z}^2}d\lambda d\omega(u_{\qu})).$
Further
\begin{equation}\label{E10a}
K_n(\qu_1,\overline{\qu_2})=\sum_{m=0}^{\infty}\frac{1}{n!m!}h_{n,m}(\qu_1,\overline{\qu}_1)
\overline{h_{n,m}(\qu_2,\overline{\qu}_2)}
\end{equation}
is a reproducing kernel of the Hilbert subspace $A_n(H)$. The convergence of the above sum easily follows from the convergence of the analogous sum in the complex case, i.e., with $\qu$ replaced by the complex variable $z$ \cite{intint}.  In particular we have $K_0(\qu,\qu)=e^{|\qu|^2}$.\\
Assume that  $\{\phi_m\}_{m=0}^{\infty}$ is an orthonormal basis of $V_H^L$. For $\qu\in H$, define
\begin{equation}\label{E10}
\eta_{\qu,n} := K_n(\qu,\overline{\qu})^{-1/2}\sum_{m=0}^{\infty}\frac{h_{n,m}(\qu,\overline{\qu})}{\sqrt{n!m!}}\phi_m\in V_H^L.
\end{equation}
The vectors $\{\eta_{\qu , n} \mid q\in H\}$ are then a family of coherent states, for each $n$. In particular, the vectors
$$
  \eta_{\qu, 0} = e^{-\frac {\vert\qu\vert^2}2} \sum_{m=0}^\infty \frac {\qu^m}{\sqrt{m!}}\phi_m
  \in V^L_H\; , $$
are the so-called {\em canonical quaternionic coherent states\/.}
The corresponding reproducing kernel, which is easily computed using (\ref{quat-repker}), is
seen to be
$$ K_0 (\qu_1 , \overline{\qu_2} ) = \sum_{n=0}^\infty \frac {\qu_1^n \; \overline{\qu_2^n}}{n!}\; , $$
the quaternionic analogue of the well-known complex Bargmann kernel. (Note that, since $\qu_1$ and
$\overline{\qu_2}$ do not necessarily commute, we cannot write this as
$\exp[\qu_1 \overline{\qu_2}]$.)

Again, the reproducing kernel Hilbert space associated to these canonical quaternionic
CS is a space of right anti-regular functons. It is of course a subspace of the bigger
space $ L_H^2(H, e^{-\abs{z}^2}d\lambda d\omega(u_{\qu}))$. This is analogous
to the space of anti-analytic functions generated by the canonical CS
(\ref{can-comp_CS}) on a complex Hilbert
space. Similarly, we could have constructed the conjugate family of
canonical quaternionic CS and the reproducing kernel space would have consisted of right
regular functions, again as in the complex case. Thus, we get the two spaces
(completion under left multiplication by quaternionic scalars is implied),
$$\mathfrak{H}_{reg}=\overline{\text{span}\left\{h_{0,m}(\qu,\overline{\qu})~\vert~m\in\mathbb{N}
\right\}}$$
 and
$$\mathfrak{H}_{a-reg}=\overline{\text{span}\left\{h_{n,0}(\qu,\overline{\qu})~\vert~n\in\mathbb{N}
\right\}}\; ,$$
of right regular and right anti-regular functions, respectively.
The first Hilbert space is the quaternionic analogue of the Bargmann space of analytic
functions, with
$$ h_{0,m} (\qu , \overline{\qu}) = \frac {\qu^m}{\sqrt{m!}}\; , \qquad m =0,1,2, \ldots\; . $$
It is interesting to note that if we define the formal annihilation and creation operators on this space
by $a h_{0,m}=  \sqrt{m}h_{0, m-1},\;\; a^\dag h_{0, m} = \sqrt{m+1}h_{0, m+1}$, then
these have realizations by the (Cullen) derivative w.r.t. $\qu$ and multiplication by $\qu$,
respectively (again in complete analogy with the complex case).

Note that since the $h_{n,m}(\qu,\overline{\qu})$ are mutually orthogonal elements in the ambient
space
$L_H^2(H,e^{-\abs{z}^2}$ $d\lambda d\omega(u_{\qu})),$ the elements of $\mathfrak{H}_{reg}$ and
$\mathfrak{H}_{a-reg}$ are mutually orthogonal except for the one-dimension\-al subspace
generated by the vector $h_{0,0}(\qu,\overline{\qu})=1$ which is common to both spaces.

\section{Conclusion}\label{sec-disc}
 Using the notion of the Cullen derivative and the related notions of regular and anti-regular functions
of quaternionic variables,
we have obtained a wide ranging generalization of certain physically interesting classes of coherent
states to quaternionic Hilbert spaces. The analysis shows, among others, that all the so-called
non-linear coherent
states, which can be realized on Hilbert spaces of analytic or anti-analytic functions, have
quaternionic generalizations. In the process we have also obtained fairly straightforward
generalizations of two different types of complex orthogonal Hermite polynomials to analogous
polynomials in a quaternionic variable, again satisfying similar orthogonality and recursion
relations. It would be interesting to explore other families of orthogonal polynomials in the same vein.
In recent years, formulations of quantum mechanics on quaternionic
Hilbert spaces have been proposed to address some of the conceptual problems associated to particle
interactions at very short distances (see, for example, \cite{Ad}). In view of the importance of
coherent states in usual quantum mechanics, it is expected that they would also be of importance in
quaternionic quantum mechanics.


\end{document}